\newcommand{\commentt}[2]{#1}
\newcommand{\comment}[1]{}

\newcommand{\cov}{{\rm Cov}}
\newcommand{\var}{{\rm Var}}

\commentt{
\newcommand{\citet}{\citeasnoun}
\documentclass[a4paper,11pt]{article}
\usepackage{graphicx,slashbox,amsthm}
\usepackage[full]{harvard}

\usepackage[ps,dvips]{xy}
\usepackage{pgfplots}
\title{On the effective dimension and multilevel Monte Carlo}
\author{Nabil Kahal\'e
\thanks{\emph{ESCP Business School,  75011 Paris,
France; {e-mail: }{nkahale@escp.eu}.}}
}
\date{\today}
\usepackage{amstext}
\usepackage{amssymb}
\usepackage{amsmath}
\usepackage{setspace}
\usepackage[margin=1in]{geometry}
\onehalfspacing
\usepackage[english]{babel}
\bibliographystyle{dcu}
\usepackage{color}
\usepackage{textcomp}
\usepackage{varioref}
\begin{document}

\newtheorem{example}{Example}[section]
\newtheorem{theorem}{Theorem}[section]
\newtheorem{conjecture}{Conjecture}[section]
\newtheorem{lemma}{Lemma}[section]
\newtheorem{proposition}{Proposition}[section]
\newtheorem{remark}{Remark}[section]
\newtheorem{corollary}{Corollary}[section]
\newtheorem{definition}{Definition}[section]
\maketitle
\newcommand{\ABSTRACT}[1]{\begin{abstract}#1\end{abstract}}
\newcommand{\citep}{\cite}
}
{
\documentclass[a4paper,11pt]{elsarticle}
\usepackage{algorithm}
\usepackage{algpseudocode}
\usepackage{hyperref}
\journal{Operations Research Letters}
\usepackage[english]{babel}
\usepackage{amssymb}
\usepackage{amsmath}
\usepackage{setspace}
\usepackage{appendix}
\usepackage{color}
\onehalfspacing
\bibliographystyle{model5-names}\biboptions{authoryear}
\begin{document}
\newtheorem{example}{Example}
\newtheorem{theorem}{Theorem}
\newtheorem{lemma}{Lemma}
\newtheorem{proposition}{Proposition}
\newtheorem{remark}{Remark}
\newtheorem{corollary}{Corollary}
\newtheorem{definition}{Definition}
\newproof{proof}{Proof}

\begin{frontmatter}
\title{On the effective dimension and multilevel Monte Carlo}
\author{Nabil Kahal\'e}
\address{ESCP Business School, 75011 Paris,
France; {e-mail: }{nkahale@escp.eu}}
}
\begin{abstract}I consider the problem of integrating  a function \(f\) over the \(d\)-dimensional unit cube. I describe a multilevel Monte Carlo method that estimates the integral with variance at most \(\epsilon^{2}\) in \(O(d+\ln(d)d_{t}\epsilon^{-2})\) time, for \(\epsilon>0\), where \(d_{t}\) is the truncation dimension of \(f\). In contrast, the standard Monte Carlo method typically achieves such variance in  \(O(d\epsilon^{-2})\) time. A lower bound of order  \(d+d_{t}\epsilon^{-2}\) is described for a class of multilevel Monte Carlo methods.
\end{abstract}
\commentt{
 }
{

\end{frontmatter}
}
\section{Introduction}
Monte Carlo simulation is used in a variety of areas including finance,  queuing systems, machine learning, and health-care.  A drawback of Monte Carlo simulation is its high computation cost. This motivates the need to design efficient simulation tools that optimize the tradeoff between the running time and the statistical error. This need is even stronger for high-dimensional problems, where the time to simulate a single run is typically proportional to the dimension. Variance reduction techniques that improve the efficiency of  Monte Carlo simulation have been developed in the previous literature  (e.g.~\cite{glasserman2004Monte,asmussenGlynn2007}). 

This paper studies the  estimation of \(\int_{[0,1]^{d}}f(x)\,dx\), where  \(f\)  is a real-valued square-integrable function on \([0,1]^{d}\). Note that \(\int_{[0,1]^{d}}f(x)\,dx=E(f(U))\),    where      \(U=(U_{1},\ldots,U_{d})\) and 
    \(U_{1},\ldots,U_{d}\)
  are independent  random variables uniformly distributed on \([0,1]\).
 The standard Monte Carlo method estimates \(E(f(U))\) by taking the average of \(f\) over \(n\) random points uniformly distributed over \([0,1]^{d}\), and achieves a statistical error of order \(n^{-1/2}\). The Quasi-Monte Carlo method (QMC) estimates 
\(E(f(U))\)  by taking the average of \(f\) over a predetermined sequence of points in \([0,1]^{d}\), and achieves   an error of order \((\log n)^{d}/n\) for certain sequences when   \(f\) has finite Hardy-Krause variation \citep[Ch.~5]{glasserman2004Monte}. Thus, for small values of \(d\), QMC can substantially outperform standard Monte Carlo.  Moreover, numerical experiments show that   QMC performs well  in certain  high-dimensional problems where  the importance of \(U_{i}\) decreases with \(i\) \citep[Ch.~5]{glasserman2004Monte}.
  \citet{CMO1997} use the ANOVA decomposition, a representation of \(f\) as the sum of orthogonal components, to define the effective dimension in the truncation sense: the truncation  dimension is low when the first variables are important. \citet{SW1998}   prove that QMC is  effective for a class of functions where high dimensions have decaying importance. The connection between QMC  and  various notions of effective dimension is studied in    \citep{EcuyerLemieux2000,owen2003,LiuOwen2006,wasilkowski2021quasi}. Methods that reduce the effective dimension and improve the performance of QMC are described in~\citep{wangSloan2011,wangTan2013,xiao2019enhancing}. \citet{owen2019effective} gives a recent survey on the effective dimension. \citet{kahaRandomizedDimensionReduction20} studies the relationship between the truncation dimension and the randomized dimension reduction method, a recent variance reduction technique applicable to high-dimensional problems.

A major advance in Monte Carlo simulation is the multilevel Monte Carlo method  
(MLMC), a variance reduction technique introduced by~\citet{Giles2008}. The MLMC method  significantly reduces the time to estimate functionals of a stochastic differential equation, and has  many other applications (e.g. \citep{Staum2017,Nobile2017,Goda2020,kahale2020Asian,Blanchet2021}).
This paper examines the connection between the MLMC method and the truncation dimension. Section~\ref{se:MLMC}
describes a MLMC method that, under suitable conditions, estimates \(E(f(U))\) with variance at most \(\epsilon^{2}\) in \(O(d+\ln(d)d_{t}\epsilon^{-2})\) time, for \(\epsilon>0\), where \(d_{t}\) is the truncation dimension of \(f\). In contrast, the standard Monte Carlo method typically achieves variance at most \(\epsilon^{2}\) in \(O(d\epsilon^{-2})\) time. My approach is based on fixing unessential variables and on approximating \(f(U)\) by functions of the first components of \(U\). Fixing unessential variables is   analysed by \citet{sobol2001global}  in the context of  the ANOVA decomposition.
Section~\ref{se:lowerBound} considers a class of MLMC estimators that approximate   \(f(U)\) by functions of the first components of \(U\).
Under general conditions, it gives a lower bound of order  \(d+d_{t}\epsilon^{-2}\) on the time required by these estimators to evaluate \(E(f(U))\) with variance at most \(\epsilon^{2}\). Section~\ref{se:examples} studies  MLMC and the truncation dimension for time-varying Markov chains with \(d\) time-steps.
Under suitable conditions, it is shown that certain  Markov chain functionals can be estimated  with variance at most \(\epsilon^{2}\) in \(O(d+\epsilon^{-2})\) time, and that the truncation dimension associated with these functionals is upper bounded by a constant independent of \(d\).  Randomized MLMC methods for equilibrium expectations of time-homogeneous Markov chains are studied in~\cite{glynn2014exact}.
\section{Preliminaries}
\subsection{The ANOVA decomposition}
It is assumed throughout the paper that \(f\) is square-integrable with    \(\var(f(U))>0\). A representation of   \(f\) in the following form:\begin{equation}\label{eq:ANOVA}
f=\sum_{Y\subseteq\{1,\dots,d\}}f_{Y},
\end{equation}  is called ANOVA decomposition if, for \(Y\subseteq\{1,\dots,d\}\) and \(u=(u_{1},\dots,u_{d})\in[0,1]^{d}\),\begin{enumerate}
\item 
 \(f_Y\) is a measurable function on \([0,1]^{d}\) and     \(f_{Y}(u)\) depends on \(u\) only through  \((u_{j})_{j\in Y}\).
 \item For    \(j\in Y\),
\begin{equation*}
\int^{1}_{0} f_{Y}(u_{1},\dots,u_{j-1},x,u_{j+1},\dots,u_{d})\,dx=0.
\end{equation*}
\end{enumerate}It can be shown~\citep[p. 272]{sobol2001global} that there is a unique ANOVA representation of \(f\), that  \(f_{\emptyset}=E(f(U))\), and  that the \(f_{Y}\)'s are square-integrable. Furthermore,  if \(Y\neq Y'\),
\begin{equation}\label{eq:ortho}\cov(f_{Y}(U),f_{Y'}(U))=0,\end{equation}and   \begin{equation}\label{eq:varf}
\var(f(U))=\sum_{Y\subseteq\{1,\dots,d\}}\sigma_{Y}^{2},
\end{equation}
where \(\sigma_{Y}\) is the standard deviation of \(f_{Y}(U)\). For \(0\leq i\leq d\), \begin{equation}\label{eq:condExpectationSobol}
E(f(U)|U_{1},\dots,U_{i})=\sum_{Y\subseteq\{1,\dots,i\}}f_{Y}(U).
\end{equation}   
\citet{owen2003} defines the truncation dimension \(d_{t}\) of \(f\)  as\begin{equation*}
d_{t}:= \frac{\sum_{Y\subseteq\{1,\dots,d\},Y\neq\varnothing}\max(Y)\sigma_{Y}^{2}}{\var(f(U))}.
\end{equation*}
For \(0\leq i\leq d\), let \begin{displaymath}
D(i):=\sum_{Y\subseteq\{1,\dots,d\},Y\neq\varnothing,\max(Y)>i}{\sigma_{Y}}^{2}
\end{displaymath} be the total variance corresponding to the last \(d-i\) components of \(f\) (see \cite{sobol2001global}).
The sequence \((D(i):0\leq i\leq d)\) is  decreasing, with   \(D(0)=\var(f(U))\) by \eqref{eq:varf} and \(D(d)=0\). Proposition~\ref{pr:varDifff} gives a bound on the variance of \(f(V)-f(V')\), when \(V\) and  \(V'\)  are  uniformly distributed on \([0,1]^{d}\) and have the same first \(i\) components. It is related to \citep[Theorem 3]{sobol2001global}. 

\begin{proposition}
\label{pr:varDifff}
Let \(i\in\{0,\dots,d\}\). Assume that \(V\) and \(V'\) are uniformly distributed on  \([0,1]^{d}\), and that \(V_{j}=V'_{j}\) for \(1\leq j\leq i\). Then \(\var(f(V)-f(V'))\le4D(i)\).
\end{proposition}
\begin{proof}As \(f_{Y}(V)=f_{Y}(V')\) for \(Y\subseteq\{1,\dots,i\}\), we have\begin{displaymath}
f(V)-f(V')=\sum_{Y\subseteq\{1,\dots,d\},Y\neq\varnothing,\max(Y)>i}f_{Y}(V)-f_{Y}(V').
\end{displaymath}
By~\eqref{eq:ortho},
\begin{displaymath}
\var\left(\sum_{Y\subseteq\{1,\dots,d\},Y\neq\varnothing,\max(Y)>i}f_{Y}(V)\right)=D(i),
\end{displaymath}
and a similar relation holds for \(V'\). Since \(\var(Z+Z')\le2(\var(Z)+\var(Z'))\) for square-integrable random variables \(Z\) and \(Z'\), this achieves the proof.
\commentt{}{\qed}\end{proof} 
Proposition~\ref{pr:LowerBoundVarAnova} gives a lower bound on the variance of the difference between \(f(U)\) and a function of the first \(i\) components of \(U\). A similar result is shown in   \citep[Theorem 1]{sobol2001global}.   
\begin{proposition}\label{pr:LowerBoundVarAnova}Let \(g\) be a real-valued square-integrable function on \([0,1]^{i}\), where \(0\leq i\leq d\).
Then \begin{equation*}
D(i)\le\var(f(U)-g(U_{1},\dots,U_{i})).
\end{equation*}\end{proposition}
\begin{proof} Define the random-variable 
\begin{equation*}\eta=f(U)-E(f(U)|U_{1},\dots,U_{i}).\end{equation*}
By properties of the conditional expectation,     \begin{equation*}
\var(\eta)\le\var(f(U)-g(U_{1},\dots,U_{i})).
\end{equation*} Combining \eqref{eq:ANOVA} and~\eqref{eq:condExpectationSobol} shows that \begin{equation*}
\eta=\sum_{Y\subseteq\{1,\dots,d\},Y\not\subseteq\{1,\dots,i\}}{f_{Y}}(U).
\end{equation*}By \eqref{eq:ortho}, \(\var(\eta)=D(i)\).
This completes the proof.
\commentt{}{\qed}\end{proof}
Proposition~\ref{pr:TruncationDimensionCharacterisation} provides an alternative characterisation of the truncation dimension.
\begin{proposition}\label{pr:TruncationDimensionCharacterisation}
\begin{displaymath}
\sum_{i=0}^{d}D(i)=d_{t}\var(f(U)).
\end{displaymath}
\end{proposition}
\begin{proof}
\begin{eqnarray*}
\sum_{i=0}^{d}D(i)
&=&\sum_{i=0}^{d}\sum_{Y\subseteq\{1,\dots,d\},Y\neq\emptyset}{\bf1}\{i<\max(Y)\}{\sigma_{Y}}^{2}\\
&=&\sum_{Y\subseteq\{1,\dots,d\},Y\neq\emptyset}\max(Y){\sigma_{Y}}^{2}\\
&=&d_{t}\var(f(U)).
\end{eqnarray*}
\commentt{}{\qed}\end{proof}
\subsection{Work-normalized variance}\label{sub:BackgroundOnWorkNormalizedVariance}

Let \(\mu\) be a real number and let \(\psi\) be a square-integrable random variable with positive variance and expected running time \(\tau\).  Assume that \(\psi\) is an unbiased estimator of \(\mu\), i.e., \(E(\psi)=\mu\). The work-normalized variance \(\tau\var(\psi)\) is a standard measure of the performance of  \(\psi\) ~\cite{glynn1992asymptotic}: asymptotically efficient unbiased estimators have low work-normalized variance.
 For \(\epsilon>0\),  let \(n_{\epsilon}\) be the smallest integer such that  the variance of the average of  \(n_{\epsilon}\)  independent copies of \(\psi\) is at most \(\epsilon^{2}\). Thus, \(n_{\epsilon}=\lceil\var(\psi)\epsilon^{-2}\rceil\). As \((x+1)/2\leq\lceil x\rceil\leq x+1\) for \(x>0\),    \begin{equation}\label{eq:TimeVarianceGen}
\frac{\tau+\tau\var(\psi)\epsilon^{-2}}{2}\le T(\psi,\epsilon)\le \tau+\tau\var(\psi)\epsilon^{-2},
\end{equation}  
where  \(T(\psi,\epsilon):=n_{\epsilon}\tau\) is the total expected time required to estimate \(\mu\) with variance at most \(\epsilon^{2}\) by taking the average of independent runs of \(\psi\).
\subsection{Reminder on MLMC}\label{sub:BackgroundOnMLMC}
Let \(\phi\) be a square-integrable random variable that is approximated with increasing accuracy by square-integrable random variables \(\phi_{l}\), \(0\leq l\leq L\), where   \(L\) is a positive integer, with \(\phi_{L}=\phi\) and \(\phi_{0}=0\). For   \(1\leq l\leq L\),   let  \(\hat\phi_{l}\) be the average of \(n_{l}\) independent copies of \(\phi_{l}-\phi_{l-1}\), where \(n_{l}\) is an arbitrary positive integer.  Suppose that     \(\hat\phi_{1},\dots,\hat\phi_{L}\) are independent. Since \begin{equation*}
E(\phi)=\sum^{L}_{l=1}E(\phi_{l}-\phi_{l-1}),
\end{equation*}     \(\hat\phi:=\sum^{L}_{l=1}\hat\phi_{l}
\) is an unbiased estimator of \(E(\phi)\), i.e.,\begin{equation}\label{eq:unbiasedMLMCGen}
E(\hat\phi)=E(\phi).
\end{equation}As observed in~\citep{Giles2008}, 
\begin{equation}\label{eq:varHatPhiGen}
\var(\hat\phi)=\sum^{L}_{l=1}\frac{V_{l}}{n_{l}},
\end{equation}where \(V_{l}:=\var(\phi_{l}-\phi_{l-1})\) for \(1\leq l\leq L\).   The expected  time required to simulate \(\hat\phi\) is \begin{equation}
\label{eq:hatTDef}
\hat T:=\sum^{L}_{l=1}n_{l}\hat t_{l},
\end{equation}where    \(\hat t_{l}\)  is the expected time to simulate \(\phi_{l}-\phi_{l-1}\).   
 The analysis in~\citep{Giles2008} shows that  
\begin{equation}\label{eq:MLMCTimeVariance}
\bigg(\sum^{L}_{l=1}\sqrt{V_{l}\hat t_{l}}\bigg)^{2}\le\hat T\var(\hat\phi),
\end{equation}  
with equality when the \(n_{l}\)'s are proportional to \(\sqrt{V_{l}/\hat t_{l}}\) (ignoring integrality constraints).
 
\section{The    MLMC algorithm}\label{se:MLMC}
Let \(L=\lceil \log_{2}(d)\rceil\)  and, for  \(0\leq l\leq L-1\),   let \(m_{l}=2^{l}-1\), with \(m_{L}=d\). For \(1\le l\leq L\) and  \(u,u'\in[0,1]^{d}\),  let\begin{equation*}
h_{l}(u,u'):=f(u_{1},\dots,u_{m_{l}},u'_{m_{l}+1},\dots,u'_{d}),
\end{equation*}   with   \(h_{0}(u,u'):=0\). Note that  \(h_{L}(u,u')=f(u)\). Let \(U'\) be a copy of \(U\) and, for \(1\leq l\leq L\), let \((U^{l,j},1\leq j\le n_{l})\) be \(n_{l}\)  copies of \(U\), where \(n_{l}:=\lceil(d/L)2^{-l}\rceil\). Assume that the random variables \((U',U^{l,j},1\leq l\leq L, 1\leq j\le n_{l})\) are independent. For \(1\leq l\leq L\),  set\begin{equation}\label{eq:defhatphi}
\tilde\phi_{l}:=\frac{1}{n_{l}}\sum ^{n_{l}}_{j=1}(h_{l}(U^{l,j},U')-h_{l-1}(U^{l,j},U')),
\end{equation}     and let \(\tilde\phi:=\sum^{L}_{l=1}\tilde\phi_{l}
\). The estimator \(\tilde\phi\) does not fall, stricto sensu, in the category of  MLMC estimators described in Section~\ref{sub:BackgroundOnMLMC}. This is because the \(n_{l}\) summands in the right-hand side of \eqref{eq:defhatphi} are dependent random variables, in general. Note that $h_{l}(u,u')$ depends on \(u\) only through its first \(m_{l}\) components. Thus, once \(U'\) is simulated,   \(h_{l}(U^{l,j},U')\) and \(h_{l-1}(U^{l,j},U')\) can be calculated  by simulating only the  first \(m_{l}\) components of \(U^{l,j}\).  For \(1\leq l\leq L\), let \(\tilde t_{l}\) be the expected time to simulate the first \(m_{l}\) components of \(U\) and  calculate \(h_{l}(U,U')\), once \(U'\) is simulated and \(f(U')\) is calculated. In other words, \(\tilde t_{l}\) is the expected time to redraw  \(U_{1},\dots,U_{m_{l}}\) and
recalculate \(f(U)\), without modifying the last \(d-m_{l}\) components of \(U\). In particular, \(\tilde t_{L}\) is the expected time to simulate \(U\) and calculate \(f(U)\).
Let   \(\tilde T\)   be the expected  time to simulate \(\tilde\phi\). 

Theorem~\ref{th:mlmcAlgPerformance} below shows that  \(\tilde\phi\) is an unbiased estimator of \(E(f(U))\). Also, when \(\hat t_{l}\) is linear in \(m_{l}\), the work-normalized variance of   \(\tilde\phi\) satisfies the bound   \(\tilde T \var(\tilde\phi)=O(\ln(d)d_{t}\var(f(U)))\), that depends on \(d\) only through \(\ln(d)\). By~\eqref{eq:MLMCTotalExpected},     \(E(f(U))\) can be estimated  via  \(\tilde\phi\) with variance at most \(\epsilon^{2}\) in expected time that depends asymptotically (as \(\epsilon\) goes to \(0\)) on \(\ln(d)\).      In contrast, assuming the expected time to simulate \(f(U)\) is of order \(d\), the work-normalized variance of the standard Monte Carlo estimator is of order \(d\var(f(U))\) and, by~\eqref{eq:TimeVarianceGen},  the standard Monte Carlo algorithm achieves variance at most \(\epsilon^{2}\) in \(O(d+d\var(f(U))\epsilon^{-2})\) expected time.    
\begin{theorem}\label{th:mlmcAlgPerformance} We have \begin{equation}\label{eq:condEx}
E(\tilde\phi)=E(\tilde\phi|U')=E(f(U)),
\end{equation}\(\var(\tilde\phi)=E(\var(\tilde\phi)|U')\), and
\begin{equation}\label{eq:varianceHatPhi}
\var(\tilde\phi)\le16\frac{\lceil \log_{2}(d)\rceil}{d}d_{t}\var(f(U)).
\end{equation}
If, for some constant \(\tilde c\) and \(1\leq l\leq L\), \begin{equation}\label{eq:LinearTimeCondition}
\tilde t_{l}\leq \tilde cm_{l},
\end{equation}   then   \(\tilde T\le9\tilde cd\) and, for \(\epsilon>0\),    \begin{equation}\label{eq:MLMCTotalExpected}
T(\tilde\phi,\epsilon)=O(d+\ln(d)d_{t}\var(f(U))\epsilon^{-2}).
\end{equation}  
 
\end{theorem}
\begin{proof} By the definition of \(\tilde\phi_{l}\),
\begin{equation*}
E(\tilde\phi_{l}|U')=E(\Delta_{l}|U'),
\end{equation*}where \(\Delta_{l}:=h_{l}(U^{},U')-h_{l-1}(U^{},U')\). Summing over \(l\) implies that \(E(\tilde\phi|U')=E(f(U))\). Taking expectations and using the tower law implies~\eqref{eq:condEx}.  Conditional on \(U'\), the \(n_{l}\) summands in the right-hand side of~\eqref{eq:defhatphi} are independent and have the same distribution as $\Delta_{l}$. Thus, for \(1\leq l\leq L\),
\begin{equation*}
\var(\tilde\phi_{l}|U')=\frac{\var(\Delta_{l}|U')}{n_{l}}.
\end{equation*}
 Furthermore, conditional on \(U'\), the random variables \(\tilde\phi_{l}\), \(1\leq l\leq L\), are independent. Hence,\begin{equation}\label{eq:condVar}
\var(\tilde\phi|U')=\sum^{L}_{l=1}\frac{\var(\Delta_{l}|U')}{n_{l}}.
\end{equation}
As \(\var(Z)=\var(E(Z|U'))+E(\var(Z|U'))\) for any square-integrable random variable \(Z\), using~\eqref{eq:condEx} shows that \(\var(\tilde\phi)=E(\var(\tilde\phi|U'))\). Similarly, \(E(\var(\Delta_{l}|U'))\leq\var(\Delta_{l})\). Consequently, taking expectations in~\eqref{eq:condVar} implies that\begin{eqnarray*}
\var(\tilde\phi)&\le&\sum^{L}_{l=1}\frac{\var(\Delta_{l})}{n_{l}}\\&\le&\frac{L}{d} \sum^{L}_{l=1}2^{l}\var(\Delta_{l}).
\end{eqnarray*}
For \(2\leq l\leq L\),  we have \(\Delta_{l}=f(V)-f(V')\), where \(V=(U_{1},\dots,U_{m_{l}},U'_{m_{l}+1},\dots,U'_{d})\),  and  \(V'=(U_{1},\dots,U_{m_{l-1}},U'_{m_{l-1}+1},\dots,U'_{d})\). Applying Proposition~\ref{pr:varDifff} with \(i=m_{l-1}\)  yields\begin{equation}\label{eq:varDeltal}
\var(\Delta_{l})\leq4D(m_{l-1}).
\end{equation}Since \(\Delta_{1}=f(U')\),    \eqref{eq:varDeltal} also holds for \(l=1\). For \(1\leq l\leq L\), we have \(2^{l}\leq 4(m_{l-1}-m_{l-2})\), where \(m_{-1}:=-1\). Hence, because the sequence \(D\) is decreasing,\begin{displaymath}
2^{l}D(m_{l-1})\leq4\sum^{m_{l-1}}_{i=m_{l-2}+1}D(i).
\end{displaymath}Thus,\begin{eqnarray*}
\sum^{L}_{l=1}{2^{l}}{\var(\Delta_{l})}
&\le&4\sum^{L}_{l=1}{2^{l}}D(m_{l-1})\\
&\le&16\sum^{L}_{l=1}\sum^{m_{l-1}}_{i=m_{l-2}+1}D(i)
\\&=&16\sum^{m_{L-1}}_{i=0}D(i)
\\&\le&16d_{t}\var(f(U)),
\end{eqnarray*}
where the last equation follows from Proposition~\ref{pr:TruncationDimensionCharacterisation}. This implies \eqref{eq:varianceHatPhi}.

Assume  now that~\eqref{eq:LinearTimeCondition} holds.  Simulating  \(\tilde\phi\) requires to draw \(U'\) and calculate \(f(U')\) once and to simulate  \(h_{l}(U,U')-h_{l-1}(U,U')\) for \(n_{l}\) independent copies of \(U\),  \(1\leq l\leq L\).
As \(m_{l}\leq2^{l}\), given \(U'\), simulating  \(h_{l}(U,U')\) (resp. \(h_{l-1}(U,U')\)) takes at most  \(\tilde c2^{l}\) (resp.  \(\tilde c2^{l-1}\)) expected time. Thus the expected time to simulate  \(h_{l}(U,U')-h_{l-1}(U,U')\)  is at most \(3\tilde c2^{l-1}\), and\begin{eqnarray*}
\tilde T&\le& \tilde cd+3\tilde c\sum^{L}_{l=1}n_{l}2^{l-1}\\
&\le& \tilde cd+3\tilde c\sum^{L}_{l=1}(1+\frac{d}{L2^{l}})2^{l-1}\\
&\le& \tilde cd+3\tilde c2^{L}+3\tilde c\frac{d}{2}\\
&\le&9\tilde cd,
\end{eqnarray*}
where the second equation follows from the inequality  \(n_{l}\le1+d/(L2^{l}\)). \eqref{eq:MLMCTotalExpected} follows immediately from~\eqref{eq:TimeVarianceGen}.
 \commentt{}{\qed}\end{proof} 
Remark~\ref{re:LinearTimeCondition} in Section~\ref{se:examples} shows that \eqref{eq:LinearTimeCondition} holds for a class of Markov chain functionals.
\subsection{Deterministic fixing of unessential variables}
The estimator \(\tilde\phi\) uses \(U'\) to fix the unessential variables. This section studies the
replacement of \(U'\) by a deterministic vector.
For  \(v\in[0,1]^{d}\) and \(1\leq l\leq L\),  set
\begin{equation*}
\tilde\phi_{l,v}:=\frac{1}{n_{l}}\sum ^{n_{l}}_{j=1}(h_{l}(U^{l,j},v)-h_{l-1}(U^{l,j},v)),
\end{equation*}  \(\tilde\phi_{v}:=\sum^{L}_{l=1}\tilde\phi_{l,v}
\). In other words, the random variable  \(\tilde\phi_{v}\) is obtained from  \(\tilde\phi\) by substituting \(U'\) with \(v\). Let   \(\tilde T(v)\)  be the expected running time of \(\tilde\phi_{v}\). For any \(v\in[0,1]^{d}\), the  estimator  \(\tilde\phi_{v}\)  falls in the class of MLMC estimators described in Section~\ref{sub:BackgroundOnMLMC}, with \(\phi=f(U)\) and \(\phi_{l}=h_{l}(U,v)\) for \(0\leq l\leq L\). Corollary~\ref{cor:fixingUnimportantVariables} shows that  $\tilde\phi_{v}$ is an unbiased estimator of \(E(f(U))\) and that there is   \(v^{*}\in[0,1]^{d}\) such that the variance of \(\tilde\phi_{v^{*}}\) and its running time are  no worse, up to a constant, than those of  \(\tilde\phi\).  
\begin{corollary}\label{cor:fixingUnimportantVariables} For   \(v\in[0,1]^{d}\), \begin{equation}\label{eq:unbiasedPhiTildeV}
E(\tilde\phi_{v})=E(f(U)).
\end{equation}Moreover, there is   \(v^{*}\in[0,1]^{d}\) such that \(\var(\tilde\phi_{v^{*}})\le3\var(\tilde\phi) \) and \(\tilde T(v^{*})\leq3\tilde T\). For \(\epsilon>0\),    \begin{equation}\label{eq:MLMCTotalExpectedv*}
T(\tilde\phi_{v^{*}},\epsilon)=O(d+\ln(d)d_{t}\var(f(U))\epsilon^{-2}).
\end{equation}   
\end{corollary}
\begin{proof}\eqref{eq:unbiasedPhiTildeV} is a special case of~\eqref{eq:unbiasedMLMCGen}. For   \(v\in[0,1]^{d}\), let   \(\xi(v):=\var(\tilde\phi_{v})\). As   \(\xi(U')=\var(\tilde\phi|U')\), it  follows from   Theorem~\ref{th:mlmcAlgPerformance} that
\(E(\xi(U'))=\var(\tilde\phi)\). Thus \(\xi(U')\le3\var(\tilde\phi)\) with probability at least \(2/3\). Similarly, \(\tilde T(U')\leq3\tilde T\) with probability at least \(2/3\). Hence,  there is   \(v^{*}\in[0,1]^{d}\) such that \(\var(\tilde\phi_{v^{*}})\le3\var(\tilde\phi) \) and  \(\tilde T(v^{*})\leq3\tilde T\). Using~\eqref{eq:TimeVarianceGen} yields \eqref{eq:MLMCTotalExpectedv*}. 
\end{proof}
The  MLMC estimator  \(\tilde\phi_{v}\)   is obtained by approximating \(f\) with functions of its first components.
A lower bound on the performance of such estimators is given in Section~\ref{se:lowerBound}.      
\section{The lower bound}\label{se:lowerBound}
This section considers a class of MLMC unbiased estimators of \(E(f(U))\)  based on successive approximations of \(f\) by deterministic functions of its first components. In \cite{kahaRandomizedDimensionReduction20}, a lower bound on the work-normalized variance of such estimators is given in terms of that of the randomized dimension reduction estimator. This section provides a lower bound on the work-normalized variance of these estimators in terms of the truncation dimension.  
  
Using the notation in Section~\ref{sub:BackgroundOnMLMC} with  \(\phi=f(U)\),  consider a MLMC estimator \(\hat \phi\) of \(E(f(U))\) obtained by summing the averages on independent copies of \(\phi_{l}-\phi_{l-1}\), \(1\leq l\leq L\), where  \(L\) is a positive integer and the   \(\phi_{l}\)'s satisfy the following assumption:
\begin{description}
\item[Assumption 1 (A1).] 
For \(0\le l\leq L\),   \(\phi_{l}\) is a square-integrable random variable equal to a deterministic measurable function of \(U_{1},\dots,U_{m_{l}}\), with \(\phi_{0}=0\) and  \(\phi_{L}=f(U)\), where \((m_{l}:0\leq l\leq L)\) is a strictly  increasing sequence of integers, with \(m_{0}=0\) and \(m_{L}=d\). \end{description}
    The proof of the lower bound is based on the following lemma.
\begin{lemma}\label{le:LowerBoundMLMC}
Let \((\nu_{i}:0\leq i\leq d)\) be a decreasing sequence such that \(\nu_{m_l}\leq\var(f(U)-\phi_{l})\) for \(0\leq l\leq L\), with \(\nu_{d}=0\). Then \begin{displaymath}
\sum_{i=0}^{d}\nu_{i}\leq\left(\sum^{L}_{l=1}\sqrt{m_{l}V_{l}}\right)^{2}.
\end{displaymath}  
\end{lemma}
\begin{proof}
An integration by parts argument \cite[ Lemma EC.4]{kahaRandomizedDimensionReduction20}  shows that
\begin{equation*}
\sum^{L-1}_{l=0}(\sqrt{m_{l+1}}-\sqrt{m_{l}})\sqrt{\nu_{m_{l}}}
\le\sum^{L}_{l=1}\sqrt{m_{l}V_{l}}.
\end{equation*} 
On the other hand, for \(0\leq l\leq L-1\), we have \begin{eqnarray*}
(\sqrt{m_{l+1}}-\sqrt{m_{l}})\sqrt{\nu_{m_{l}}}&=&\sum_{i=m_{l}}^{m_{l+1}-1}(\sqrt{i+1}-\sqrt{i})\sqrt{\nu_{m_{l}}}\\&\ge&\sum_{i=m_{l}}^{m_{l+1}-1}\alpha_{i},
\end{eqnarray*}
where \(\alpha_{i}=(\sqrt{i+1}-\sqrt{i})\sqrt{\nu_{i}}\). Summing over \(l\in\{0,\dots,L-1\}\) implies that 
\begin{equation*}
\sum^{d}_{i=0}\alpha _{i}\le\sum^{L}_{l=1}\sqrt{m_{l}V_{l}}.
\end{equation*}
On the other hand, \begin{eqnarray*}\left(\sum^{d}_{i=0}\alpha_{i}\right)^{2}&=&\sum^{d}_{i=0}\alpha_{i}\left(\alpha_{i}+2\sum^{i-1}_{j=0}\alpha_{j}\right)\\
&\ge&\sum^{d}_{i=0}\alpha_{i}\left(\alpha_{i}+2\sum^{i-1}_{j=0}(\sqrt{j+1}-\sqrt{j})\sqrt{\nu_{i}}\right)\\
&=&\sum^{d}_{i=0}\alpha_{i}(\alpha_{i}+2\sqrt{i\nu_{i}})\\
&=&\sum^{d}_{i=0}\nu_{i}.
\end{eqnarray*}This concludes the proof.
\commentt{}{\qed}\end{proof} 
Theorem~\ref{th:multilevelLowerBound} provides a lower bound the work-normalized variance of \(\hat\phi\) that matches, up to a logarithmic factor, the upper bound in Theorem~\ref{th:mlmcAlgPerformance}. \begin{theorem}\label{th:multilevelLowerBound}
If Assumption A1 holds and there is a positive constant \(\hat c\) such that   \(\hat t_{l}\geq \hat c m_{l}\) for \(1\leq l\leq L\), then \(\hat cd_{t}\var(f(U))\leq \hat T\var(\hat\phi)\) and, for \(\epsilon>0\),    \begin{equation}\label{eq:MLMCTotalExpectedLowerBound}
T(\tilde\phi,\epsilon)=\Omega(d+d_{t}\var(f(U))\epsilon^{-2}).
\end{equation}
 \end{theorem}
\begin{proof}
  
It follows from~\eqref{eq:MLMCTimeVariance} that\begin{displaymath}
\hat c\left(\sum^{L}_{l=1}\sqrt{m_{l}V_{l}}\right)^{2}\le \hat T\var(\hat\phi).
\end{displaymath} By Proposition~\ref{pr:LowerBoundVarAnova} and Assumption A1,  \(D({m_l})\leq\var(f(U)-\phi_{l})\) for \(0\leq l\leq L\). Applying Lemma~\ref{le:LowerBoundMLMC} with \(\nu_{i}=D(i)\) for \(0\leq i\leq d\) yields  
\begin{eqnarray*}\left(\sum^{L}_{l=1}\sqrt{m_{l}V_{l}}\right)^{2}
&\ge&\sum_{i=0}^{d}D(i)\\
&=&d_{t}\var(f(U)),
\end{eqnarray*}
where the second equation follows from Proposition~\ref{pr:TruncationDimensionCharacterisation}.
This shows that \(\hat cd_{t}\var(f(U))\leq \hat T\var(\hat\phi)\). By~\eqref{eq:hatTDef}, the expected running time of \(\hat\phi\) is lower-bound by \(\hat t_{L}\ge\hat c d\). Together with~\eqref{eq:TimeVarianceGen}, this implies \eqref{eq:MLMCTotalExpectedLowerBound}. \commentt{}{\qed}\end{proof} 
\section{Time-varying Markov chains}
\label{se:examples}
This section shows that, under certain conditions, the expectation of functionals of time-varying Markov chains with \(d\) time-steps can be estimated efficiently via MLMC, and that the associated truncation dimension is upper bounded by a constant independent of \(d\).

 Let  \(d\) be a positive integer and let \((X_i:0\leq i\leq d)\) be a time-varying Markov chain with state-space
\(F\) and deterministic initial value \(X_{0}\).
Assume  that  there are independent  random variables \(Y_{i}\),  \(0\leq i\leq d-1\), uniformly distributed in \([0,1]\),  and measurable
functions \(g_{i}\) from \(F\times [0,1]\) to \(F\) such
that \(X_{i+1}=g_{i}(X_{i},Y_{i})\) for \(0\leq i\leq d-1\).     Our goal is to estimate  \(E(g(X_{d}))\)
where \(g\) is a deterministic real-valued
measurable function on \(F\) such that \(g(X_{d})\) is square-integrable. It is assumed that \(g\) and the \(g_{i}\)'s
can be calculated in constant time.
  For \(1\leq i\leq d\), set \(U_{i}=Y_{d-i}\). An inductive argument shows that there is a real-valued measurable function \(f\) on \([0,1]^d\)
such that \(g(X_{d})=f(U)\), where \(U=(U_{1},\dots,U_{d})\). When \(X_{d}\) is mainly determined by the last \(Y_{i}\)'s,  the first \(U_{i}\)'s are the most important arguments of \(f\).
\begin{remark}\label{re:LinearTimeCondition}
Redrawing  \(U_{1},\dots,U_{i}\) while keeping  \(U_{i+1},\dots,U_{d}\)   unchanged amounts to keeping \(X_{0},\dots,X_{d-i}\) unchanged  and  redrawing \(X_{d-i+1},\dots, X_{d}\). This can be achieved in \(O(i)\) time. Thus \eqref{eq:LinearTimeCondition} holds for \(f\).
\end{remark}
Given \(i\in\{0,\dots,d\}\), define the time-varying Markov chain \((X^{(i)}_{j}:d-i\leq j\leq d)\) by setting \(X^{(i)}_{d-i}:=X_{0}\) and \(X^{(i)}_{j+1}=g_{j}(X^{(i)}_{j},Y_{j})\) for \(d-i\leq j\leq d-1\). Thus,  \(X^{(i)}_d\) is the state of the original Markov chain \(X\) at time-step \(d\)  if  the chain is at state \(X_{0}\) at time-step \(d-i\).  Note that  \(g(X^{(i)}_{d})\)   can be calculated in \(O(i)\) time and is a deterministic function of \(U_{1},\dots,U_{i}\). Roughly speaking, if \(X_{d}\) is determined to a large extent by the last \(Y_{j}\)'s, then \(X^{(i)}_{d}\) should be ``close'' to \(X_{d}\) for large values of \(i\). This motivates the following assumption:
\begin{description}
\item[Assumption 2 (A2).] 
There are constants \(c'\) and \(\gamma<-1\) independent of \(d\) such that, for \(0\leq i\leq d\), we have \(
E((g(X_{d})-g(X^{(i)}_{d}))^{2})\leq c'(i+1)^{\gamma}\).
\end{description}

I now describe a multilevel estimator of \(E(\phi)\), where \(\phi=g(X_{d})\), using the notation in Section~\ref{sub:BackgroundOnMLMC}. Let \(L=\lceil \log_{2}(d)\rceil\)  and, for  \(1\leq l\leq L-1\),   let \(m_{l}=2^{l}-1\). Let $\phi_{0}=0$,   \(\phi_{L}=d\) and, for   \(1\leq l\leq L-1\),   let \(\phi_{l}=g(X^{(m_{l})}_{d})\). For   \(1\leq l\leq L\),   let  \(\hat\phi_{l}\) be the average of \(n_{l}\) independent copies of \(\phi_{l}-\phi_{l-1}\), where \(n_{l}=\lceil d2^{l(\gamma-1)/2}\rceil\).  Suppose that     \(\hat\phi_{1},\dots,\hat\phi_{L}\) are independent. Set \(\hat\phi:=\sum^{L}_{l=1}\hat\phi_{l}
\). By~\eqref{eq:unbiasedMLMCGen}, \(E(\hat\phi)=E(\phi)\). Let \(\hat T\) (resp. \(\hat t_{l}\)) be the expected  time to simulate \(\hat\phi\) is (resp. \(\phi_{l}-\phi_{l-1}\)). Proposition~\ref{pr:MarkovChains} shows that, under Assumption A2, \(\hat\phi\) can be used to estimate \(E(\phi)\) with precision \(\epsilon\) in \(O(d+\epsilon^{-2})\) time and, if  \(\var(g(X_{d}))\) is lower-bounded by a constant independent of \(d\), the truncation dimension \(d_{t}\) associated with  \(g(X_{d})\) is upper-bounded by a constant independent of \(d\). In contrast, the standard Monte Carlo method typically achieves precision \(\epsilon\) in \(O(d\epsilon^{-2})\) time.
\begin{proposition}\label{pr:MarkovChains}
Suppose that Assumption A2 holds.
Then there are constants \(c_{1}\), \(c_{2}\) and \(c_{3}\) independent of \(d\) such that  \(\hat T\leq c_{1}d\),   \(\var(\hat\phi)\leq c_{2}/d\), and \(T(\hat\phi,\epsilon)\leq c_{3}(d+\epsilon^{-2})\). Moreover,\begin{displaymath}
d_{t}\var(g(X_{d}))
\le c'\frac{\gamma}{\gamma+1}.   
\end{displaymath}
\end{proposition}
\begin{proof}
By construction,  \(\hat t_{l}\le c2^{l}\) for some constant \(c\) independent of \(d\). By~\eqref{eq:hatTDef},
\begin{eqnarray*}
\hat T
&\le&c\sum^{L}_{l=1}(1+ d2^{l(\gamma-1)/2})2^{l}\\
&\le&cd(4+\frac{1}{1-2^{(\gamma+1)/2}}).
\end{eqnarray*}By~Assumption A2, for   \(0\leq l\leq L\), \begin{displaymath}
\var(g(X_{d})-\phi_{l})\leq c'2^{l\gamma}.
\end{displaymath}
Since \(\var(Z+Z')\le2(\var(Z)+\var(Z'))\) for square-integrable random variables \(Z\) and \(Z'\), it follows that  \(V_{l}\leq4 c'2^{(l-1)\gamma}\) for \(1\leq l\leq L\). 
Together with~\eqref{eq:varHatPhiGen}, this shows that 
\begin{eqnarray*}
\var(\hat\phi)
&\le&4 c'\sum^{L}_{l=1}\frac{2^{(l-1)\gamma}}{ d2^{l(\gamma-1)/2}}\\
&\le&\frac{4 c'2^{(1-\gamma)/2}}{d(1-2^{(\gamma+1)/2})}.
\end{eqnarray*}
Using \eqref{eq:TimeVarianceGen} implies the desired bound on \(T(\hat\phi,\epsilon)\).

By Proposition~\ref{pr:LowerBoundVarAnova}, for \(0\leq i\leq d\), 
\begin{eqnarray*}
D(i)&\leq& E((g(X_{d})-g(X^{(i)}_{d}))^{2})\\
&\le& c'(i+1)^{\gamma}.
\end{eqnarray*}
Thus, using Proposition~\ref{pr:TruncationDimensionCharacterisation},
\begin{eqnarray*}
d_{t}\var(f(U))&\le& c'\sum^{d}_{i=1}i^{\gamma}\\
&\le&c'(1+\int^{d}_{1}x^{\gamma}\,dx)\\
&\le&c'\frac{\gamma}{\gamma+1}.
\end{eqnarray*}

\commentt{}{\qed}\end{proof}
\subsection{A Lindley recursion example}\label{sub:singleQueue}
In this example, \(F=\mathbb{R}\) and   \((X_i:0\leq i\leq d)\) satisfies the time-varying
Lindley equation \begin{equation*}
X_{i+1}=(X_{i}+\zeta_{i}(Y_{i}))^{+},
\end{equation*} with \(X_{0}=0\), where  \(\zeta_{i}\),  \(0\leq i\leq d-1\), is a real-valued function on \([0,1]\). Our goal is to estimate \(E(X_{d})\). Thus \(g\) is the identity function and \(g_{i}(x,y)=(x+\zeta_{i}(y))^{+}\) for \((x,y)\in\mathbb{R}\times[0,1]\). Lindley equations often arise in queuing theory~\cite{asmussenGlynn2007}.
\begin{proposition}\label{pr:GD1}If  there are constants \(\theta>0\) and \(\kappa<1\) independent of \(d\) such that
\begin{equation}\label{eq:momentCondGD1}
E(e^{\theta\zeta_{i}(Y_{i})})\leq\kappa
\end{equation}  for \(0\leq i\leq d-1\), then  \(E((X_{d}-X^{(i)}_{d})^{2})\leq \theta'\kappa^{i}\) for \(0\leq i\leq d-1\), where \(\theta'\) is a constant independent of \(d\). 
\end{proposition}

Proposition~\ref{pr:GD1} shows that, if  \eqref{eq:momentCondGD1} holds, then so does Assumption A2, hence the conclusions of Proposition~\ref{pr:MarkovChains} hold as well. The proof of Proposition~\ref{pr:GD1} is essentially the same as that of \cite[Proposition 10]{kahaRandomizedDimensionReduction20}, and is therefore omitted.   A justification of \eqref{eq:momentCondGD1} for time-varying queues and  numerical examples showing the efficiency of  MLMC  for estimating Markov chain functionals are given in \cite{kahaRandomizedDimensionReduction20}.  
\section*{Acknowledgments}  This work was achieved through the Laboratory of Excellence on Financial Regulation (Labex ReFi) under the reference ANR-10-LABX-0095. It benefitted from a French government support managed by the National Research Agency (ANR). The author thanks Art Owen for helpful comments.
\bibliography{poly}
\end{document}